\documentclass[runningheads,a4paper,envcountsame]{llncs}
\usepackage{graphicx}
\usepackage{amsmath, amssymb, color}

\usepackage{graphicx}
\usepackage{epsfig}
\usepackage{graphics}
\usepackage{array}
\usepackage{amsmath}
\usepackage{amsfonts}
\usepackage{amssymb}
\usepackage[algo2e,ruled,vlined,linesnumbered]{algorithm2e}
\usepackage{setspace}
\usepackage{fullpage} 

\usepackage{color}
\definecolor{orange}{rgb}{1,0.5,0}

\newtheorem{observation}[theorem]{Observation}

\newcommand{\N}{\mathbb{N}}
\newcommand{\R}{\mathbb{R}}

\newcommand{\M}{\mathcal{M}}
\newcommand{\B}{\mathcal{B}}
\newcommand{\Ce}{\mathcal{C}}
\newcommand{\Es}{\mathcal{S}}
\newcommand{\EE}{\mathcal{E}}

\renewcommand{\phi}{\varphi}

\DeclareMathOperator{\size}{size}
\DeclareMathOperator{\lsp}{lsp}
\DeclareMathOperator{\short}{sp}

\DeclareMathOperator{\MCB}{MCB}
\DeclareMathOperator{\LSC}{LSC}
\DeclareMathOperator{\blue}{blue}
\DeclareMathOperator{\green}{green}
\DeclareMathOperator{\decomp}{decomp^*}
\DeclareMathOperator{\expand}{expand^*}
\DeclareMathOperator{\de}{decomp}
\DeclareMathOperator{\ex}{expand}
\DeclareMathOperator{\shrink}{shrink}
\DeclareMathOperator{\dfs}{dfs}

\DeclareMathOperator{\GF}{GF}
\DeclareMathOperator{\degree}{deg}

\begin{document}
\mainmatter  
\title{Computing Minimum Cycle Bases in Weighted Partial 2-Trees in Linear Time\thanks{This is the full version of~\cite{WG} which has appeared in the proceedings of the 39th International Workshop on Graph-Theoretic Concepts in Computer Science (WG 2013).}}
\titlerunning{Computing Minimum Cycle Bases in Weighted Partial 2-Trees in Linear Time}

\author{
Carola Doerr\inst{1,2}\thanks{Work initiated while the author was with Universit\'e Paris Diderot - Paris 7, LIAFA, Paris, France and the Max Planck Institute for Informatics.},
G. Ramakrishna\inst{3}\thanks{Work initiated while visiting the Max Planck Institute for Informatics during an internship.},
Jens M. Schmidt\inst{4}
}
\authorrunning{Doerr, Ramakrishna, Schmidt}

\institute{
Sorbonne Universit\'es, UPMC Univ Paris 06,  UMR 7906, LIP6, F-75005 Paris, France\and
CNRS, UMR 7906, LIP6  F-75005 Paris, France\and
Indian Institute of Technology Madras, India\and
Max Planck Institute for Informatics, Saarbr\"ucken, Germany
}

\toctitle{Computing Minimum Cycle Bases in Weighted Partial 2-Trees in Linear Time}
\tocauthor{Doerr, Ramakrishna, Schmidt}
\maketitle

\sloppy{ 
\begin{abstract}
We present a linear time algorithm for computing an implicit linear space representation of a minimum cycle basis (MCB) in weighted partial 2-trees; i.e., graphs of treewidth two. 
The implicit representation can be made explicit in a running time that is proportional to the size of the MCB.

For planar graphs, Borradaile, Sankowski, and Wulff-Nilsen [Min $st$-cut Oracle for Planar Graphs with Near-Linear Preprocessing Time, FOCS 2010] showed how to compute an implicit $O(n \log n)$ space representation of an MCB in $O(n \log^5 n)$ time. For the special case of partial 2-trees, our algorithm improves this result to linear time and space. Such an improvement was achieved previously only for outerplanar graphs [Liu and Lu: Minimum Cycle Bases of Weighted Outerplanar Graphs, IPL 110:970--974, 2010].
\end{abstract}
}

\sloppy{
\section{Introduction}

A \emph{cycle basis} of a graph $G$ is a minimum-cardinality set $\Ce$ of cycles in $G$ such that every cycle $C$ in $G$ can be written as the exclusive-or sum of a subset of cycles in $\Ce$. A \emph{minimum cycle basis} (MCB) of $G$ is a cycle basis that minimizes the total weight of the cycles in the basis.
Minimum cycle bases have numerous applications in the analysis of electrical networks, biochemistry, periodic timetabling, surface reconstruction, and public transportation, and have been intensively studied in the computer science literature. We refer the interested reader to~\cite{Kavitha} for an exhaustive survey.
It is---both from a practical and a theoretical viewpoint---an interesting task to compute minimum cycle bases efficiently.

All graphs considered in this work are simple graphs $G=(V,E)$ with a non-negative edge-weight function $w:E \rightarrow \R_{\geq 0}$. (Computing MCBs for graphs with cycles of negative weight is an NP-hard problem~\cite{Kavitha}. In all previous work that we are aware of it is therefore assumed that the edge-weights are non-negative.) Throughout this work, $m=|E|$ denotes the size of the edge set and $n=|V|$ the size of the vertex set of $G$.

\subsection{Previous Work}

The first polynomial-time algorithm for computing MCBs was presented by Horton in 1987~\cite{Horton87}. His algorithm has running time $O(m^3 n)$.
This was improved subsequently in a series of papers by different authors, cf.~\cite{Kavitha} or~\cite{LiuL10} for surveys of the history.
The currently fastest algorithms for general graphs are a deterministic $O(m^2n / \log n)$ algorithm of Amaldi, Iuliano, and Rizzi~\cite{Amaldi10MCB} and a Monte Carlo based algorithm by Amaldi, Iuliano, Jurkiewicz, Mehlhorn, and Rizzi~\cite{AmaldiIJMR09} of running time $O(m^{\omega})$, where $\omega$ is the matrix multiplication constant. 

The algorithm from~\cite{AmaldiIJMR09} is deterministic on planar graphs, and has a running time of $O(n^2)$. 
This improved the previously best known bound by Hartvigsen and Mardon~\cite{HartvigsenM94}, which is of order $n^2 \log n$. 
The currently best known algorithm on planar graphs is due to Borradaile, Sankowski, and Wulff-Nilsen~\cite{FOCS10}. It constructs an $O(n \log n)$ space implicit representation of an MCB in planar graphs in time $O(n \log^5 n)$.\footnote{\textbf{Note added in proof.} 
The arxiv paper~\cite{WulffArxiv} announces an improved running time of $O(n \log^4 n)$.}

Faster algorithms for planar graphs are known only for the special case of outerplanar graphs. For unweighted outerplanar graphs, Leydold and Stadler~\cite{LeydoldS98} presented a linear time algorithm. 
More recently, Liu and Lu~\cite{LiuL10} presented a linear time, linear space algorithm to compute an MCB of a weighted outerplanar graph (using an implicit representation). This is optimal both in terms of time and space. 

\subsection{Our Result}
\label{sec:ourresult}
In this contribution, we consider the computation of minimum cycle bases of partial 2-trees. The class of partial 2-trees (also referred to as graphs of treewidth two) is a strict superclass of outerplanar graphs. It includes 2-trees and series-parallel graphs (as shown in~\cite[Theorem 42]{Bodlaender98}, a graph $G$ is a partial 2-tree if and only if
every 2-connected component of $G$ is a series-parallel graph).
Partial 2-trees are planar. 
They are precisely the graphs that forbid a $K_{4}$-subdivision. 

Partial 2-trees are extensively studied in the computer science literature, e.g., a deterministic logspace algorithm is presented to canonize and test isomorphism for partial 2-trees~\cite{Arvind08}; plane embeddings of partial 2-trees are described in~\cite{Proskurowski96}; parallel strategies can be used to find the most vital edges~\cite{Ho98}; and the oriented chromatic number of partial 2-trees is studied in~\cite{Ochem08}. 
Partial 2-trees can be recognized in linear time~\cite{Bodlaender96}.

Our main result is a linear time algorithm for computing an implicit $O(n)$-space representation of a minimum cycle basis in partial 2-trees. The explicit representation can be obtained in additional time that is proportional to the size of the MCB. 
Since partial 2-trees are planar graphs, the previously best known algorithm was the one by Borradaile, Sankowski, and Wulff-Nilsen~\cite{WulffArxiv}. That is, for the special case of partial 2-trees we are able to improve their running time by a factor of $\Theta(\log^4 n)$.

Our result is achieved by an iterative decomposition of the partial 2-tree into outerplanar graphs to which the recent result of Liu and Lu~\cite{LiuL10} can be applied. We state our main theorem below. 
As will be discussed in Section~\ref{sec:conclusions} the ideas presented here do not carry over to planar $3$-trees. 
It thus seems that substantially new ideas are required to improve the running time of~\cite{WulffArxiv} for graph classes containing graphs of treewidth at least three. 

\begin{theorem}
\label{thm:main}
Given a partial $2$-tree $G=(V,E)$ on $n$ vertices and a non-negative weight function $w: E \rightarrow \R_{\geq 0}$, 
a minimum cycle basis $\B$ of $G$ (implicitly encoded in $O(n)$ space) can be obtained in $O(n)$ time. 

Moreover, $\B$ can be reported explicitly in time $O(\size(\B))$, where $\size(\B)$ is the number of edges in $\B$ counted according to their multiplicity.
\end{theorem}

Note in Theorem~\ref{thm:main} that although $\B$ has an \emph{implicit} representation of linear size, the \emph{explicit size} of $\B$ may be quadratic. This is true already for outerplanar graphs, cf.~\cite{LiuL10} for a simple ladder graph $G$ in which the unique MCB of $G$ contains $\Theta(n^2)$ edges.

For the proof of Theorem~\ref{thm:main} it will be crucial that the set of \emph{lex short cycles} (cf. Section~\ref{sec:lex}) in any weighted partial 2-tree forms a minimum cycle basis~\cite{CTW12}. As lex short cycles are inherently defined by shortest paths, we will need a data structure that reports the distance between two vertices in constant time (e.g., for checking whether an edge is the shortest path between its two endpoints). In outerplanar graphs, such a data structure exists due to Frederickson and Jannardan~\cite{FredericksonJ88}. For our more general case, we will instead extend a result of Chaudhuri and Zaroliagis~\cite[Lemma~3.2]{ChaudhuriZ00} on weighted partial $k$-trees which supports distance queries between two vertices of a common bag in a (fixed) tree decomposition of $G$ in constant time.
Using this extension, we can report a shortest path $P$ between any two such vertices in time $O(|E(P)|)$.  

\section{Graph Preliminaries, Partial 2-Trees, and Lex Shortest Paths}
\label{sec:definitions}

We consider weighted undirected graphs $G = (V,E)$ where $V$ denotes the set of vertices, $E$ the set of edges, and $w: E \rightarrow \mathbb{R}_{\geq 0}$ a non-negative weight function. We assume the usual unit-cost RAM as model of computation (real RAM when the input contains reals), but we do not use any low-level bit manipulations in our algorithm to improve the running time.
All graph classes considered in this paper are sparse, i.e., we have a linear dependence $m = O(n)$.

The weight $w(P)$ of a path $P$ in $G$ is the sum of weights of edges in $P$;
i.e., $w(P) := \sum_{e \in P}{w(e)}$.
A set $X \subset V$ of vertices in $G$ is said to be a \emph{vertex separator} of $G$ if the removal of the vertices $X$ increases the number of connected components. A vertex separator $X$ is \emph{minimal} if no proper subset of $X$ is a vertex separator. 
For $Y \subseteq V$, we write $G[Y]$ for the subgraph of $G$ that is induced by $Y$ and we write $G-Y$ for the subgraph that is obtained from $G$ by deleting the vertices in $Y$ (and all incident edges).
For $k, \ell \in \N$, we denote by $K_{\ell}$ the complete graph on $\ell$ vertices and by $K_{\ell,k}$ we denote the the complete bipartite graph on $\ell$ and $k$ vertices.
For a graph $H$, an $H$-subdivision is a graph obtained from $H$ by replacing its edges with non-empty and pairwise vertex-disjoint paths. In this work, we will be mainly concerned with $K_{2,k}$-subdivisions for $k \geq 3$. In such a $K_{2,k}$-subdivision we call the two vertices of degree greater than two the \emph{branch vertices} of the subdivision.

\subsection{Minimum Cycle Bases}
\label{sec:MCBs}
A cycle $C$ in $G$ is a connected subgraph of $G$ in which every vertex has degree two. 
Let $C_1,\ldots,C_k$ be cycles in $G$ and let $\oplus$ denote the symmetric difference function. Then the sum 
$\Es := C_1 \oplus \ldots \oplus C_k$ is the set of
edges appearing an odd number of times in the multi-set $\{C_1, \ldots ,C_k\}$. As is well known, $\Es$ is a union of cycles in $G$.

We say that a set $\Ce = \{C_1,\ldots ,C_k\}$ of cycles of $G$ \emph{spans the cycle space of
$G$} if every cycle $C$ of $G$ can be written as a sum $C_{i_1} \oplus \ldots \oplus C_{i_{\ell}}$ of elements
of $\Ce$. In this case, we say that $C_{i_1}, \ldots, C_{i_{\ell}}$ \emph{generate}~$C$.
The size $\emph{size}(\Ce)$ of $\Ce$ is the number of edges in $\Ce$ counted according to their multiplicity.

A \emph{cycle basis} of $G$ is a minimum cardinality set of cycles that spans the cycle space of $G$. Put differently, a cycle basis is a maximal set of independent cycles, where we consider a set of cycles to be independent if their incidence vectors in $\{0,1\}^m$ are independent over the field $\GF(2)$.
The cardinality of a cycle basis is sometimes referred to as the \emph{dimension} of the cycle space of $G$.
The dimension of the cycle space of any simple weighted graph equals $m-n+1$~\cite{Bondy2008}.

We are interested in identifying a \emph{minimum cycle basis} ($\MCB$) of $G$; i.e., a cycle basis $\Ce$ of minimum total weight $\sum_{C \in \Ce}{w(C)}$. 

If $G_1, \ldots, G_k$ are the $2$-connected components of the graph $G$ and if $\Ce_1, \ldots, \Ce_k$ are minimum cycle bases of $G_1, \ldots, G_k$, respectively, then the union $\Ce_1 \cup \ldots \cup \Ce_k$ is a minimum cycle basis of $G$.
In what follows, we will therefore assume without loss of generality that $G$ is $2$-connected.

\subsection{Tree Decompositions and Partial 2-Trees}

A \emph{tree decomposition} of a graph $G$ is a pair $(\{X_1, \ldots, X_r\}, T)$ of a set of \emph{bags} $X_1, \ldots, X_r$ and a tree $T$ with vertex set $V(T)=\{X_1, \ldots, X_r\}$ that satisfies the following three properties:
\begin{enumerate}
\item $X_1 \cup \ldots \cup X_r = V$,
\item For each edge $\{u,v\} \in E$, there is an index $1 \leq i \leq r$ such that $\{u,v\}\subseteq X_{i}$, and 
\item For each vertex $v \in V$, the bags in $T$ containing $v$ form a subtree of $T$ (\emph{subtree property}).
\end{enumerate}
The \emph{treewidth} of $(\{X_1, \ldots X_r\}, T)$ is $\max \{|X_1|, \ldots |X_r|\}-1$. 
The \emph{treewidth} of $G$ is the minimum treewidth over all possible tree decompositions of $G$. We call a tree decomposition $(\{X_1, \ldots X_r\}, T)$ \emph{optimal} if the treewidth of $T$ is equal to the treewidth of $G$. To distinguish between the edges of $G$ and $T$, we refer to the edges of $T$ as \emph{links}. 

A \emph{$k$-tree} is a graph of treewidth $k$ for which the addition of any edge between non-adjacent vertices would increase the treewidth. The following lemma is folklore and characterizes $k$-trees.

\begin{lemma}
\label{constructivecharacterization}
A graph $G$ is a $k$-tree if and only if $G$ can be constructed from a $K_{k+1}$ by iteratively adding new vertices such that the neighborhood of each such vertex is a $k$-clique.
\end{lemma}

A subgraph of a $k$-tree is called a \emph{partial $k$-tree}. 
As mentioned in the introduction, partial 2-trees form a strict superclass of outerplanar graphs. In fact, while outerplanar graphs are characterized by the forbidden minor set $\{K_4,K_{2,3}\}$, partial $2$-trees have the forbidden minor set $\{K_4\}$. 
Equivalently, a partial 2-tree is outerplanar if and only if it does not contain a $K_{2,3}$-subdivision (as a subgraph).
The following statement is taken from~\cite{CTW12}. 

\begin{lemma}[Lemma 2.4 in~\cite{CTW12}]
\label{lem:K23minor}
A partial 2-tree $G$ is not outerplanar if and only if $G$ contains a $K_{2,3}$-subdivision. If $G$ contains a $K_{2,k}$-subdivision for $k \geq 3$, each of its $k$ non-branch vertices is contained in a separate connected component of $G-\{u,v\}$; in particular, $G-\{u,v\}$ has $k \geq 3$ connected components.
\end{lemma}

%

\subsection{Lex Shortest Paths and Lex Short Cycles}\label{sec:lex}
It is known (Proposition 4.5 in~\cite{HartvigsenM94}) that for any edge-weighted simple graph $G$ the set of so-called \emph{lex short cycles} contains a minimum cycle basis. For outerplanar graphs~\cite{LiuL10} and partial 2-trees~\cite{CTW12}, the whole set of lex short cycles forms a minimum cycle basis. 

\begin{definition}[Lex Shortest Paths]
\label{def:lex}
Let $G=(V,E)$ be a graph with weight function $w:E \rightarrow \R_{\geq 0}$. Let $\sigma: V \rightarrow \{1,2, \ldots, n\}$ be an arbitrary ordering of the vertices.

A path $P$ between two distinct vertices $u,v \in V$ is called a \emph{lex shortest path} if for any other path $P'$ between $u$ and $v$ either
$w(P')>w(P)$ or 
\emph{(}$w(P')=w(P)$ and $|E(P')| > |E(P)|$\emph{)} or 
\emph{(}$w(P')=w(P)$, $|E(P')| = |E(P)|$ and $\min_{y \in V(P')\setminus V(P)} \sigma(y) > \min_{y \in V(P)\setminus V(P')} \sigma(y)$\emph{)} holds. 
\end{definition}

It is easily verified that between any two vertices $u,v$ in $G$ there exists exactly one lex shortest path. We refer to this path as $\lsp(u,v)$ (cf. also Proposition 4.1 in~\cite{HartvigsenM94}). If the dependence of the graph is not clear from the context, we write $\lsp_G(u,v)$. Note that every subpath of a lex shortest path is a lex shortest path.

\begin{definition}[Lex Short Cycles]
A \emph{lex short cycle} $C$ is a cycle that contains for any two vertices $u,v \in C$ the lex shortest path $\lsp(u,v)$. For an edge-weighted graph $G$, we denote by $\LSC(G)$ the set of all lex short cycles in $G$.
\end{definition}

\begin{lemma}[\cite{HartvigsenM94,CTW12}]
\label{lem:LSCMCB} 
For any edge-weighted simple graph $G$, there is a set $\B \subseteq \LSC(G)$ such that $\B$ is a minimum cycle basis for $G$. 
Additionally, the set of lex short cycles $\LSC(G)$ forms a minimum cycle basis if $G$ is a weighted partial 2-tree.
\end{lemma}

Abusing notation (since $\MCB(G)$ may not be unique) we write $\MCB(G) \subseteq \LSC(G)$ and $\MCB(G) = \LSC(G)$, respectively, for the two statements in Lemma~\ref{lem:LSCMCB}.

By Lemma~\ref{lem:LSCMCB} it would suffice to compute the set of lex short cycles in $G$ for our purposes. 
This is the approach of Liu and Lu, who showed that for outerplanar graphs an implicit representation of $\LSC(G)$ can be computed in linear time. 
Before commenting further on our algorithm, we briefly note that in their paper, Liu and Lu assume that the shortest paths in their outerplanar graph are unique. They motivate this assumption by introducing a preprocessing routine that perturbs the input weights accordingly (cf. Lemma 1 in~\cite{LiuL10}). However, it is not obvious how to run this preprocessing step in linear time (as is misleadingly stated there), since the weight differences that are needed become exponentially small. Therefore, arithmetic operations on these numbers cannot be done in constant time; the unit-cost assumption was never meant to be stretched that far. However, it turns out that this assumption is not needed, as their algorithm relies only on the following fact, which we briefly note and prove here for the sake of completeness.

\begin{lemma}[implicitly in~\cite{LiuL10}]\label{Lemma1}
Let $G=(V,E)$ be a weighted outerplanar graph such that all the edges in 
$G$ are the lex shortest paths between their two endpoints .
A cycle $C$ in $G$ is a lex short
cycle if and only if $C$ is an induced cycle in $G$.
\end{lemma}

\begin{proof}
Let $C'$ be a non-induced cycle in $G$. Then there exists two 
non-neighboring vertices $x',y'$ in $C'$ with $\{x',y'\} \in E$. Since 
all edges in $G$ are the shortest paths between their endpoints, by definition, we have 
$\lsp(x',y')=\{x',y'\}$. Since $x'$ and $y'$ are non-neighboring in 
$C'$, the lex shortest path between these two vertices is not in $C'$. 
Hence, $C'$ cannot be lex short.

Let us now consider a cycle $C$ in $G$ that is not lex-short. To show 
that $C$  cannot be an induced cycle, we assume the contrapositive and 
show that this contradicts the outerplanarity of $G$.
Let $x$ and $y$ be two vertices in $C$ whose lex shortest path 
$\lsp(x,y)=(x=x_1, \ldots, x_t=y)$ is not contained in $C$.
That is, there at least one index $i$ such that the edge $\{x_i, 
x_{i+1}\}$ is not in $C$. Let $p$ be a minimal such index.
Since $C$ is an induced cycle and every edge in $G$ is the shortest path between its two endpoints, the vertex 
$x_{p+1}$ cannot lie on the cycle, i.e., $x_{p+1} \notin C$. Let $q$ be 
the smallest index greater than $p$ such that $x_q \in C$.
We observe that $x_p$ and $x_q$ cannot be neighboring each other in $C$, 
for otherwise $x_1, \ldots, x_p, x_q, \ldots, x_t$ would be a path 
between $x$ and $y$ of length strictly smaller than $\lsp(x,y)$.
This shows that there exists a $K_{2,3}$-subdivision in $G$.
$G$ can thus not be outerplanar by Lemma~\ref{lem:K23minor}.
\qed
\end{proof}

From this (and an implementation of a result in~\cite{FredericksonJ88}) is not too difficult to see that an implicit representation of the set of lex short cycles of an outerplanar graph can be obtained in linear time, cf.~\cite{LiuL10} for details.

\begin{theorem}[\cite{LiuL10}]
\label{thm:LiuL}
For every weighted outerplanar graph $G$ on $n$ vertices an $O(n)$-space representation of the set $\LSC(G)$ can be computed in $O(n)$ time. From this representation, any cycle $C \in \LSC(G)$ can be computed explicitly in time $O(\size(C))$. 
\end{theorem}

As mentioned above, if we could compute for every pair of vertices $u$ and $v$ in $G$ the lex shortest path between $u$ and $v$, we could---like Liu and Lu---resort to computing $\LSC(G)$. 
However, since we do not know currently how to perturb the weights in linear time such that all lex shortest paths in $G$ are also the unique shortest paths between its two endpoints, it is somewhat more challenging to achieve a linear running time for our generalization. Our algorithm will therefore not necessarily compute the set of lex short cycles. Instead, as we shall describe in the next section, we will do a decomposition of the graph $G$ into outerplanar graphs using arbitrary shortest paths instead of lex shortest paths. 
The result of Liu and Lu will still be an essential step in our algorithm, as it allows us to handle the outerplanar graphs that result from our decomposition using Theorem~\ref{thm:LiuL}.

\section{High-Level Overview of Our Algorithm and Technical Details}
\label{subsec:highlevel}

We first describe the high-level idea of our algorithm; most proofs and the algorithmic details are presented in the subsequent sections. From now on we assume that $G$ is a 2-connected weighted partial 2-tree.

\subsection{Removal of Long Edges}
\label{sec:preprocessingSteps}

We call an edge $\{u,v\}$ \emph{tight} if it is a shortest path between $u$ and $v$, and we call it \emph{long} otherwise. By the observation made in the following lemma, we will treat the long edges in $G$ separately. In fact, the lemma allows us to ignore the set $L$ of all long edges in the main routine of our algorithm. For two vertices $u$ and $v$ in $G$, let $\short(u,v)$ be an arbitrary shortest path between $u$ and $v$.

\begin{lemma}
\label{lem:tight}
Let $G=(V,E)$ be a weighted partial 2-tree and let $L$ be the set of long edges in $G$. 
Then $\MCB(G) = \MCB(G\setminus L) \cup \{\{e\} \cup \short(u,v) \mid e=\{u,v\} \in L\}$.
\end{lemma}
 
\begin{proof}
The independence of $\M:=\MCB(G\setminus L) \cup \{ \{e\} \cup \short(u,v) \mid e=\{u,v\} \in L\}$ follows from the independence of $\MCB(G\setminus L)$ and the fact that $e$ is contained in $\M$ only in the one cycle $\{e\} \cup \short(u,v)$. 
Since $\M$ contains exactly $m-n+1$ cycles, it is (cf. the comment in Section~\ref{sec:MCBs}) also a \emph{maximal} set of independent cycles; i.e., a cycle basis.
We verify that the total weight of the cycles in $\M$ is minimal: According to Lemma~\ref{lem:LSCMCB}, $\LSC(G \setminus L)$ is a minimum cycle basis of $G \setminus L$. Thus, the weight of the cycles in $\MCB(G\setminus L)$ equals the weight of the cycles in $\LSC(G \setminus L)$. In addition, the weight of each cycle $\{e\} \cup \short(u,v)$ equals the weight of $\{e\} \cup \lsp(u,v)$.  It follows directly from the two arguments given above that $\LSC(G)=\LSC(G\setminus L) \cup \{ \{e'\} \cup \lsp(u',v') \mid e'=\{u',v'\} \in L\}$. Since $\LSC(G)$ is a minimum cycle basis of $G$, $\M$ is of minimum weight. 
\qed
\end{proof}

According to Lemma~\ref{lem:tight}, we can ignore long edges, but need to ensure that we add a short cycle containing $e$ for every edge $e \in L$ to the minimum cycle basis at the very end of the main routine. That long edges can be identified and removed from $G$ in $O(n)$ time using a suitable data structure will be shown in Lemma~\ref{lem:DistanceOracleBoost}. The removal of the long edges from $G$ does therefore not change the linear runtime of our algorithm. 

\subsection{High-Level Overview of the Main Algorithm}
\label{sec:highlevel}

In a first step of Algorithm~\ref{alg:main} we remove the set of long edges $L$ from $G$ (they will be taken care of later on using Lemma~\ref{lem:tight}). 
The key approach for our algorithm is then to iteratively decompose the graph $G \setminus L$ into outerplanar graphs $\tilde{G}_1, \ldots, \tilde{G}_r$. 
To these graphs we apply the linear time algorithm of Liu and Lu (Theorem~\ref{thm:LiuL}). 
Intuitively, the decomposition is done as follows.

When $G \setminus L$ is not outerplanar, then there exists a $K_{2,3}$-subdivision in $G \setminus L$ with branch vertices $u$ and $v$ such that 
(i) $\{u,v\}$ is a minimum vertex separator of $G \setminus L$ and 
(ii) the removal of $\{u,v\}$ disconnects $G \setminus L$ into at least three connected components $H_1, \ldots, H_k$ (cf. Lemma~\ref{lem:K23minor}).
We distinguish two cases. 
If $\{u,v\} \in E$, we set $G_h:=(G \setminus L)[V(H_h) \cup \{u,v\}]$, $1 \leq h \leq k$. 
Otherwise, let 
$\short(u,v)$ be an arbitrary shortest path between $u$ and $v$ in $G \setminus L$ and let 
$j(u,v) \in \{1, 2, \ldots, k\}$ such that 
$\short(u,v) \in (G \setminus L)[V(H_{j(u,v)}) \cup \{u,v\}]$. 
 
We set 
$G_{j(u,v)}:=(G \setminus L)[V(H_{j(u,v)}) \cup \{u,v\}]$, and for all 
$1 \leq h \neq j(u,v) \leq k$ we set 
$G_h:=(G \setminus L)[V(H_h) \cup \{u,v\}] \cup \green(u,v)$, where $\green(u,v)$ denotes a new ``colored'' (i.e., marked) edge $\{u,v\}$ that serves as a placeholder for the shortest path $\short(u,v)$ between $u$ and $v$ (which, by definition, is not contained in $G_h$). 
The weight $w(\green(u,v))$ assigned to this new edge is therefore set to the weight $w(\short(u,v))$ of the shortest path between $u$ and $v$. Clearly, $w(\green(u,v))=w(\lsp(u,v))$.
Let the operation $\de(G \setminus L,u,v)$ decompose $G \setminus L$ into $G_1, \ldots, G_k$ with respect to the vertices $u,v$. We call each $G_h$, $1 \leq h \leq k$, a \emph{part} of $\de(G \setminus L,u,v)$.

We now iteratively decompose the graphs $G_1, \ldots, G_k$ as described above until we are left with graphs $\tilde{G}_1, \ldots, \tilde{G}_r$ that do not contain any $K_{2,3}$-subdivision, i.e., with outerplanar graphs according to Lemma~\ref{lem:K23minor}. Since all the edges in $\tilde{G}_h$ are tight, the set of lex short cycles in $\tilde{G}_h$ equals the boundaries of its internal faces. Extracting the internal faces of $\tilde{G}_h$ can be done in linear time, cf.~\cite{LiuL10} or the comments before and after Lemma~\ref{Lemma1}. 
We will show in Theorem~\ref{thm:mapping1} that the (disjoint) union of $\ex(\LSC(\tilde{G}_1)), \ldots, \ex(\LSC(\tilde{G}_r))$ forms a minimum cycle basis, where, naturally, $\ex(\LSC(G_h))$ replaces the marked edges $\green(u,v)$ in every cycle by the shortest path $\short(u,v)$.

Finally, we add to this minimum cycle basis the cycles $e \cup \short(u,v)$ for all long edges $e=\{u,v\} \in L$ in $G$, where again $\short(u,v)$ is an arbitrary shortest path between $u$ and $v$. 
This can be done either implicitly by storing $u$, $v$, and the graph $G$ or explicitly by computing the shortest paths with Lemma~\ref{lem:DistanceOracleBoost}.

An important part of the algorithm is to find a data structure that allows to identify all $K_{2,3}$-subdivisions and to do the respective decomposition in linear time. 
To this end, we define \emph{suitable tree decompositions}. 

\subsection{Suitable Tree Decompositions}
\label{sec:treedecomp}

We define suitable tree decompositions and we show how they help in efficiently computing our decomposition. To ease readability, we consider \emph{rooted} tree decompositions. 
We direct all links in the tree decomposition from the root to the leaves, that is, for a link $(X,Y)$ in the decomposition, $X$ has smaller distance to the root than $Y$. Bag $X$ will then be referred to as the \emph{father}, and bag $Y$ is referred to as the \emph{child}. 
All links $\ell =(X, Y)$ in the tree decomposition are \emph{labeled} by the intersection $X \cap Y$ of its two endpoints.

\begin{definition}[Suitable Tree Decomposition]
\label{def:suitable}
An optimal rooted tree decomposition of $G$ is \emph{suitable} if it satisfies the following properties:
\begin{enumerate}
	\item The size of every bag $X_i$ is $3$ and every two adjacent bags $X_i, X_j$ in $T$ differ by exactly one vertex; i.e., $|X_i \cap X_j|=2$ (this property is called \emph{smooth} in~\cite{Bodlaender96}).
	\item Any two links with the same label have a common father in $T$; i.e., for any two links $(X_1,Y_1)$ and $(X_2,Y_2)$ with $X_1 \cap Y_1 = X_2 \cap Y_2$ it holds that $X_1 = X_2$.
\end{enumerate}
\end{definition}
Observe that for any internal bag in $T$, the number of children could be arbitrary, but there are at most three different labels associated with the links to its children.

Our algorithm will perform all computations in a suitable tree decomposition of the tight induced subgraph of $G$. It is therefore important that such a tree decomposition can be computed in linear time.

\begin{lemma}
\label{lem:makesuitable}
Given a partial 2-tree $G$, a suitable tree decomposition can be computed in linear time and has linear space.
\end{lemma}
\begin{proof}
The number of bags of a smooth tree decomposition of a partial 2-tree $G$ is $n-2$ (see Lemma~2.5 in~\cite{Bodlaender96}). As every bag contains exactly $3$ vertices, this gives a linear space representation of the tree decomposition.
A rooted tree decomposition $(\{X_1, \ldots, X_r\}, T)$ of $G$ that is additionally smooth can be computed using Bodlaender's algorithm~\cite{Bodlaender96}. We make $T$ suitable (i.e., we add Property~\ref{def:suitable}.2) as follows. By the subtree-property of tree decompositions, it suffices to give a smooth tree decomposition $T'$ in which no two links $(A,B)$ and $(B,C)$ have the same label.

Traverse $T$ in any order that starts on the root and in which father bags precede their children. Whenever a bag $B$ with father $A$ and child $C$ is visited such that the links $(A,B)$ and $(B,C)$ have the same labels, we modify $T$ by attaching the subtree of $T$ that is rooted on $C$ to $A$; i.e., after the modification $C$ is a sibling of $B$. This causes $B$ and $C$ to have the same father. It is straight-forward to see that the modified tree is still a smooth tree decomposition. After the traversal is finished, we have a suitable tree decomposition. Clearly, all modifications can be computed in constant time per step, leading to a total running time of $O(n)$.
\qed
\end{proof}

One of the key observations of our algorithm is the fact that for all $K_{2,3}$-subdivisions the two branch vertices must be contained in at least \emph{three common bags} of a suitable tree decomposition. This is shown using the following results (cf.\ Corollary~\ref{cor:graphTreeDecomp2}).

\begin{lemma}[Lemma~12.3.4 in~\cite{Diestel2010}]
\label{lem:Diestel}
Let $W \subseteq V(G)$ and let $T$ be a tree decomposition of $G$. Then $T$ contains either a bag that contains $W$ or a link $(X_1,X_2)$ such that two vertices of $W$ are separated by $X_1 \cap X_2$ in $G$.
\end{lemma}

\begin{lemma}[$K_{2,3}$-Subdivisions in Partial 2-Trees]
\label{lem:K23minorInBag}
Let $u$ and $v$ be the branch vertices of a $K_{2,3}$-subdivision $H$ in a partial 2-tree~$G$. For every optimal tree decomposition $T$ of $G$ without bag duplicates (in particular for suitable tree decompositions), $\{u,v\}$ is contained in at least one bag of $T$.
\end{lemma}

\begin{proof}
We show the claim by applying Lemma~\ref{lem:Diestel} with $W = \{u,v\}$. If $u$ and $v$ are not contained in a bag, there must be a link $(X_1,X_2)$ in $T$ such that $u$ and $v$ are separated by $X_1 \cap X_2$ in $G$. Since $H$ is a $K_{2,3}$-subdivision, at least three vertices need to be removed in order to separate $u$ and $v$. Since $T$ is optimal, $X_1 \cap X_2$ can only contain more than two vertices when $X_1$ and $X_2$ consist of the same three vertices. This contradicts that there are no bag duplicates in $T$.
\qed\end{proof}

Now we can prove the desired Corollary~\ref{cor:graphTreeDecomp2} with the following lemma.

\begin{lemma}
\label{obs:graphTreeDecomp1}
Let $T$ be a suitable tree decomposition of a $2$-connected partial 2-tree $G$ and let $u,v \in V$. Then $T$ contains at least three bags that contain both $u$ and $v$ if and only if $G-\{u,v\}$ has at least three connected components. If $T$ contains at least $k \geq 3$ such bags, the number of connected components in $G-\{u,v\}$ is exactly $k$; in particular, $G$ contains then a $K_{2,k}$-subdivision.
\end{lemma}

\begin{proof}
In the following, let $Y_1,\ldots,Y_k$ be the bags in $T$ containing both $u$ and $v$. By the subtree property of tree decompositions, $Y_1,\ldots,Y_k$ induce a connected subgraph in $T$. Since $T$ is suitable, we can assume that $Y_2,\dots,Y_k$ are children of $Y_1$. Let $F$ be the forest obtained from $T$ by deleting the links $(Y_1,Y_2),\ldots,(Y_1,Y_k)$. For each $1 \leq i \leq k$, let $T_i$ be the subtree in $F$ containing $Y_i$ and let $V_i = \{x \in V \mid \text{ there is a bag containing } x \text{ in } T_i\}$.

Let $T$ contain $k \geq 3$ bags that contain both $u$ and $v$. We prove that $G-\{u,v\}$ has at least $k$ connected components. Let $i,j \in \{1,\ldots,k\}$ with $i \neq j$. By the subtree property of tree decompositions and $Y_i \neq Y_j$, we have $V_i \cap V_j = \{u,v\}$. Since all bags in $T$ contain exactly three vertices, the sets $V_i - V_j$ and $V_j - V_i$ are non-empty. We need to show that for all $x \in V_i - V_j$ and all $y \in V_j - V_i$ there is no edge $xy$ in $G$. Assume to the contrary that such an edge $xy$ exists. Then there is a bag $B$ in $T$ that contains both $x$ and $y$. Since $x \in V_i - \{u,v\}$ and by the subtree property, $T_i$ is the only tree with a bag containing $x$. Similarly, $T_j$ is the only tree with a bag containing $y$. This contradicts the existence of $B$. Thus, every of the $V_i - \{u,v\}$ is the vertex set of a connected component of $G-\{u,v\}$.

Let $G -\{u,v\}$ have $\ell \geq 3$ connected components. We prove that $T$ contains at least $\ell$ bags that contain both $u$ and $v$. It is well-known that every connected component that is obtained by deleting a minimal vertex separator is adjacent to all vertices of this separator. Since $G$ is $2$-connected, $\{u,v\}$ is a minimal vertex separator of $G$. It follows that $G$ contains a $K_{2,\ell}$-subdivision $H$ with branch-vertices $u$ and $v$ and non-branch-vertices $x_1,\ldots,x_\ell$ different from $u$ and $v$. According to Lemma~\ref{lem:K23minorInBag}, at least one bag of $T$ contains both $u$ and $v$. Thus, there is at least one $Y_i$, $T_i$ and $V_i$ defined as above (note that $T_1 = T$ and $V_1 = V$ if $Y_i$ is the only bag containing $u$ and $v$). We prove that no set $V_i$ contains two vertices $x_a$ and $x_b$ with $1 \leq a < b \leq \ell$. This gives the claim, as it implies that there are at least $l$ trees $T_i$ and thus, at least $l$ bags $Y_i$, each of which contains $u$ and $v$. By construction of $T_i$, $T_i$ has exactly one bag $C$ that contains $u$ and $v$ (note that the uniqueness of $C$ exploits the fact that $T$ is suitable). Let $c$ be the vertex $C \setminus \{u,v\}$. Assume to the contrary that $T_i$ contains a bag $A$ containing $x_a$ and a bag $B$ containing $x_b$ ($A = B$ is possible). Let $L$ be the least common ancestor of $A$ and $B$ in $T_i$. We know from the existence of $H$ that there are two independent paths from $x_a$ to $u$ and $v$, respectively, that share only the vertices $u$ and $v$; similarly, there are two such independent paths from $x_b$ to $u$ and $v$, respectively. It follows that $L \neq C$ and, in particular $c \notin \{x_a,x_b\}$, since otherwise two of the four independent paths from $\{x_a,x_b\}$ to $\{u,v\}$ would intersect in $c \notin \{u,v\}$. However, if $L \neq C$, two of these independent paths must intersect in $c$ as well in order to reach $u$ and $v$, which gives the desired contradiction.
\qed\end{proof}

From Lemmata~\ref{lem:K23minor} and~\ref{obs:graphTreeDecomp1}, we obtain the following corollary.

\begin{corollary}
\label{cor:graphTreeDecomp2}
Let $T$ be a suitable tree decomposition of a $2$-connected partial $2$-tree $G$. Then $G$ is outerplanar if and only if for every two nodes $u,v \in V$ at most two bags of $T$ contain $u$ and $v$.
\end{corollary}

Corollary~\ref{cor:graphTreeDecomp2} allows us to efficiently find all $K_{2,k}$-subdivisions for $k \geq 3$ in a $2$-connected partial $2$-tree by finding $k$ pairwise adjacent bags in a suitable tree decomposition that share the same two vertices $u$ and $v$.

\subsection{Suitable Data Structures for Finding the Lex Shortest Paths}
\label{sec:datastructure}

Another useful tool in our algorithm will be the following data structure. 
It supports the query for an intermediate vertex that lies on a shortest path between two nodes. The following lemma is along the lines of~\cite{FredericksonJ88}. 

\begin{lemma}
\label{Lemma_routingInTrees}
Let $T$ be an unrooted tree decomposition of $G$. There is a linear space data-structure with $O(n)$ preprocessing time that supports the following query: Given a bag $A \in T$ and a vertex $v$ in $G$ that is not in $A$, find the link incident to $A$ that leads to some bag containing $v$. The query time is $O(\log d)$, where $d$ is the degree of $A$ in $T$.
\end{lemma}

\begin{proof}
Note that the desired link is unique, as $T$ is a tree decomposition. For building the data structure, we perform a depth first search ($\dfs$) on $T$, starting at an arbitrary artificial root, and label every bag $X$ with a dfs-number. We label each vertex of a bag $X$ with the dfs-number of $X$. For any bag $X$ and the subtree $T(X)$ of $T$ that is rooted at $X$, the bags in $T(X)$ get consecutive dfs-numbers. Hence, these numbers form an interval, which we can store in constant space at $X$ during the depth first search; i.e., in linear total time. Similarly, the bags not in $T(X)$ get dfs-numbers that are consecutive in the cyclic order of dfs-numbers; we store the corresponding interval at the father bag of $X$ (if exists). The desired answer for the query is then obtained by performing a binary search on the neighboring bags of $A$ (performed in the same order as the dfs) that stops at the bag having an interval that contains the label of $v$. This takes time $O(\log d)$. 
\qed
\end{proof}

We are finally ready to show that for any long edge $\{u,v\}$ a shortest path $\short(u,v)$ between $u$ and $v$ can be computed in time $O(|E(\short(u,v))|)$.
The following lemmata will be useful also to identify the subtree of the tree decomposition that contains $\short(u,v)$ for two branch vertices $u$ and $v$ with $\{u,v\} \notin E$.

\begin{lemma}[Lemma~3.2 in~\cite{ChaudhuriZ00}]
\label{lem:DistanceOracle}
Given a partial $k$-tree $G$ and an optimal tree decomposition $T$ of
$G$, there is an algorithm with running time $O(k^3n)$ that outputs the
distances of all vertex pairs that are contained in common bags and
that, for each such vertex pair, outputs some intermediate vertex of a shortest path between the vertices.
\end{lemma}

Lemma~\ref{lem:DistanceOracle} is originally stated for directed graphs
in~\cite{ChaudhuriZ00}. However, representing each undirected edge with
two edges oriented in opposite directions gives the above undirected
variant. 

We extend Lemma~\ref{lem:DistanceOracle} by giving the following data
structure. 

\begin{lemma}\label{lem:DistanceOracleBoost}
Given a connected partial $2$-tree $G$ and a suitable tree decomposition
$T$ of $G$, there is an $O(n)$-space data structure requiring $O(n)$
preprocessing time that supports the following queries, given two vertices $u$ and $v$ and a bag $X \in T$ that contains
$u$ and $v$:
\begin{itemize}
	\item Compute in time $O(1)$ the length of a shortest path
between $u$ and $v$ \emph{(distance query)}.
	\item Compute in time $O(1)$ an intermediate vertex $w$ of some shortest path between $u$ and $v$, and a bag $Y \in T$ such that $Y=\{u,v,w\}$, providing that any shortest path between $u$ and $v$ has at least two edges \emph{(intermediate vertex query)}.
	\item Compute in time $O(|E(P)|)$ a shortest path $P$ between
$u$ and $v$ \emph{(shortest path extraction)}.
\end{itemize}
\end{lemma}

Since a tree decomposition maintains for every edge the bag that contains it, the queries of Lemma~\ref{lem:DistanceOracleBoost} can in particular be performed when---instead of the bag $X$---an edge $\{u,v\} \in G$ is given.

\begin{proof}[of Lemma~\ref{lem:DistanceOracleBoost}]
We apply the algorithm of Lemma~\ref{lem:DistanceOracle} and store the
distance of every vertex pair $\{u,v\}$ that is contained in a common
bag, say in $X$, in a table linked to $X$. Since $T$ contains only
linearly many bags, this takes $O(n)$ space. The table supports distance
queries in constant time, as there are only constantly many vertex pairs
in each bag.

Assume for the moment that we know how to support the intermediate vertex query. Then we can easily support the shortest path extraction by first applying an intermediate
vertex query, which gives $Y$, and subsequently recursing on the two intermediate vertex queries $\{u,w\}$ and $\{w,v\}$, both in $Y$, until each shortest path is just an edge. This allows to extract a shortest path between $u$ and $v$ in time proportional to its length.

It remains to show how to support intermediate vertex queries. We initialize the data structure $D$ of
Lemma~\ref{Lemma_routingInTrees} for the tree decomposition $T$ in time
$O(n)$ and apply the algorithm of Lemma~\ref{lem:DistanceOracle} in time
$O(n)$. Let $X$ be a bag containing $u$ and $v$. By Lemma~\ref{lem:DistanceOracle}, we have already found an intermediate vertex $z$
between $u$ and $v$, but want to find an intermediate vertex $w$ that is in a common bag $Y$ with $u$ and $v$.
If $z$ does not exist, there is a shortest path that is just an edge, in which case we just set $w$ to be non-existent as well. If $z \in X$, we set $w = z$ and $Y=X$ and
are done.

Otherwise, we query $D$ with $(X,z)$ and get a link $(X,A)$ such that $z$ is contained in the subtree of $T$ that is separated by $(X,A)$ and does not contain $X$ (note that $A$ may be the father of $X$ in $T$). According to Lemma~\ref{Lemma_routingInTrees}, this query takes time proportional to at most the degree of $X$ in $T$.

We now distinguish two cases. In the case that $A$ contains $u$ and $v$,
we iterate this procedure on $A$ instead on $X$. In this iteration, this case cannot
happen more than a constant number of times, as $T$ is suitable, so any
path in the subtree of $T$ consisting of bags containing $\{u,v\}$ has
length at most $2$.

Otherwise, $A$ contains exactly one vertex of $\{u,v\}$, say $u$. Consider $X=\{u,v,r\}$ and the subtree $T_1$ of $T$ that is separated by the link $(X,A)$ and contains $A$. By the subtree property, $T_1$ cannot contain a bag with $v$, as then $v$ would also be contained in $A$. Since $T_1$ contains a part of a shortest path between $u$ and $v$, but has only $u$ and $r$ in common with $X$, $r$ must be an intermediate vertex. Since $X$ contains $u$, $v$, and $r$, we set $w = r$ and $Y = X$.

We investigate the preprocessing time of the data structure, i.e., the time spent computing for all vertex pairs $(u,v)$ the intermediate vertex $w$ and the bag containing all three vertices $\{u,v,w\}$.
In every bag $X$, there are only constantly many vertex pairs. For each such vertex pair, we could find $w$ in time $O(\degree(X))$, where $\degree(X)$ is the degree of $X$ in $T$. Hence, the
preprocessing time sums up to a linear total.
\qed
\end{proof}

\subsection{Obtaining $\MCB(G)$ from $\LSC(\tilde{G}_1), \ldots, \LSC(\tilde{G}_r)$}
\label{sec:graphTechnicalities}

As a last technicality, we show that---as claimed in the high-level overview of our algorithm---the disjoint union $\ex(\LSC(\tilde{G}_1)) \uplus \ldots \uplus \ex(\LSC(\tilde{G}_r))$ forms a minimum cycle basis of $G \setminus L$. 

Recall that $G_{j(u,v)}$ is the part of $\de(G \setminus L,u,v)$ containing the shortest path $\short(u,v)$ between $u$ and $v$ along which we have decomposed $G \setminus L$. 

\begin{definition}
For any cycle $C$ of $G \setminus L$, let $\ex(C)$ be the cycle obtained from $C$ by replacing the green edges $\green(u,v)$ in $C$ (if exist) by the shortest path $\short(u,v)$ between $u$ and $v$ in the part $G_{j(u,v)}$. 
For a set of cycles $\Ce$, let $\ex(\Ce):=\{\ex(C) \mid C \in \Ce\}$.
\end{definition}

\begin{theorem}
\label{thm:mapping1}
$\MCB(G \setminus L) = \ex(\LSC(\tilde{G}_1)) \uplus \ldots \uplus \ex(\LSC(\tilde{G}_r))$.
\end{theorem}

Theorem~\ref{thm:mapping1} follows from iteratively applying the following lemma.
\begin{lemma}
\label{lem:mapping}
Let $G$ be a graph in which every edge is tight. 
Let $u$ and $v$ be the two branch vertices of a $K_{2,3}$-subdivision in $G$. 
Let $G_1, \ldots, G_k$ be the subgraphs resulting from the decomposition $\de(G,u,v)$. For each $1 \leq h \leq k$, let $\B_h$ be a minimum cycle basis of the graph $G_h$. 
Then 
$\EE:=\ex(\B_1) \cup \ldots \cup \ex(\B_k)$ is a minimum cycle basis of $G$.
\end{lemma}

To prove Lemma~\ref{lem:mapping} we first introduce an alternative decomposition, $\decomp(G,u,v)$, which decomposes a non-outerplanar graph with respect to the \emph{lex short path} between the two branch vertices---as opposed to the decomposition $\de(G,u,v)$ which decomposes $G$ along an arbitrary shortest path.

Similarly to the decomposition $\de(G,u,v)$ described in Section~\ref{sec:highlevel} let $G$ be a graph that is not outerplanar and let $u,v \in V$ be the branch vertices of a $K_{2,3}$ subdivision in $G$. Let $H_1, \ldots, H_k$ be the connected components of $G-\{u,v\}$.

\textbf{Case 1: } If $\{u,v\} \in E$, set 
$G^*_h:=G_h=G[V(H_h) \cup \{u,v\}]$, $1 \leq h \leq k$. 

\textbf{Case 2: } If $\{u,v\} \notin E$, let $i(u,v) \in \{1, 2, \ldots, k\}$ such that 
$\lsp(u,v) \in G[V(H_{i(u,v)}) \cup \{u,v\}]$. 
Set 
$G^*_{i(u,v)}:=G[V(H_{i(u,v)}) \cup \{u,v\}]$, and for all 
$1 \leq h \neq {i(u,v)} \leq k$ we set 
$G^*_h:=G[V(H_h) \cup \{u,v\}] \cup \blue(u,v)$, where $\blue(u,v)$ is a marked edge $\{u,v\}$ that serves as a placeholder for the lex shortest path $\lsp(u,v)$ between $u$ and $v$ (which is not contained in $G^*_h$).
Set $w(\blue(u,v))=w(\lsp(u,v))$ 
For any cycle $C$ of $G$, let $\expand(C)$ be the cycle obtained from $C$ by replacing the blue edges $\blue(u,v)$ in $C$ (if exist) by the lex shortest path $\lsp(u,v)$. 
For a set of cycles $\Ce$, let $\expand(\Ce):=\{\expand(C) \mid C \in \Ce\}$.

The proof of Lemma~\ref{lem:mapping} is based on the following result, which we believe to be of independent interest.
  
\begin{lemma}
\label{lem:lscdecomp} 
In the setting of Lemma~\ref{lem:mapping} let $G^*_1, \ldots, G^*_k$ be the subgraphs resulting from the decomposition $\decomp(G,u,v)$. 
Then $\LSC(G) = \expand(\LSC(G^*_1)) \uplus \ldots \uplus \expand(\LSC(G^*_k))$.
\end{lemma}

Lemma~\ref{lem:lscdecomp} can be proven using the following observation. 
\begin{observation}[Lemma 2.5 and Corollary 2.8 in~\cite{CTW12}]
\label{lem:lspobservations}
Let $G$ be a weighted graph and let $G'$ be a subgraph of $G$. Let $P$ be a path in $G'$.
If $P$ is lex shortest in $G$, it is lex shortest in $G'$.

Furthermore, for $G$, $k$, and $G^*_1, \ldots, G^*_k$ as in Lemma~\ref{lem:lscdecomp}, we have $\expand(\LSC(G^*_h)) \subseteq \LSC(G)$, $1 \leq h \leq k$.
\end{observation}

\begin{proof}[of Lemma~\ref{lem:lscdecomp}]
The disjointness of the sets follows immediately from the facts that only cycles are contained in these sets and that the subgraphs $H_1, \ldots, H_k$ in $G-\{u,v\}$ are disjoint.
The inclusion $\bigcup_{h=1}^k \expand(\LSC(G^*_h)) \subseteq \LSC(G)$ follows from Observation~\ref{lem:lspobservations}.

It thus remains to show $\LSC(G) \subseteq \bigcup_{h=1}^k \expand(\LSC(G^*_h))$. 
To this end, let $C \in \LSC(G)$. 
We need to show that $E(C)\setminus E(\lsp_{G}(u,v))$ is contained in one of the $G^*_h$; Observation~\ref{lem:lspobservations} implies that for any such cycle either $C$ itself or the cycle $\shrink(C)$ with the lex shortest path $\lsp_{G}(u,v)$ replaced by the edge $\blue(u,v)$ must be contained in $\LSC(G^*_h)$.

Since the decomposition is done along the vertices $u$ and $v$, there is nothing to show in case $|C \cap \{u,v\}| \leq 1$. 
Indeed, any such $C$ or its short version $\shrink(C)$ is contained in exactly one connected component $G^*_h$.

Let us therefore assume that both vertices $u$ and $v$ are contained in $C$. Since $C$ is a lex short cycle in $G$, it must contain the lex shortest path $\lsp(u,v)$ between $u$ and $v$. 
The cycle $C$ is complemented by another path $P$ from $u$ to $v$; i.e., there exists a path $P$ from $u$ to $v$ such that $C=\lsp(u,v) \cup P$ and $V(P) \cap V(\lsp(u,v))=\{u,v\}$. By the structure of our decomposition, this path $P$ is certainly contained in one connected component $G^*_h$. Since $G^*_h$ contains also either $\lsp(u,v)$ itself or the placeholder edge $\blue(u,v)$ we have that either $C$ or $\shrink(C)$ is contained in $G^*_h$. As mentioned above, by Observation~\ref{lem:lspobservations} it follows that $C \in \expand(\LSC(G^*_h))$. 
\qed
\end{proof}

Before we are finally ready to prove Lemma~\ref{lem:mapping}, we observe that from Lemma~\ref{lem:K23minor} it follows that in the parts $G_{j(u,v)}$ and $G^*_{i(u,v)}$ there are no two vertex-disjoint paths between $u$ and $v$.

\begin{observation}
\label{cor:nocycle}
If a partial 2-tree $G$ contains a $K_{2,3}$-subdivision with branch vertices $\{u,v\}$, the subgraph $G_{j(u,v)}$ defined by $\de(G,u,v)$ and the subgraph $G^*_{i(u,v)}$ defined by $\decomp(G,u,v)$ do not contain a cycle that contains both vertices $u$ and $v$.
\end{observation}

\begin{proof}[of Lemma~\ref{lem:mapping}]
For $1 \leq h \leq k$ let 
$m_h$ be the number edges in the graph $G_h$ and let 
$n_h$ be the number of vertices in $G_h$.
By the observation made in Section~\ref{sec:MCBs} the number of cycles in $\B_h$ equals $m_h-n_h+1$.
The number of cycles in $\EE$ therefore equals 
\begin{align*}
\sum_{h=1}^k{(m_h-n_h+1)} = m+k-1 - (n+2(k-1)) +k = m-n+1.
\end{align*}
We therefore need to show that the cycles in $\EE$ are independent and that they are of minimum total weight.

As for the independence assume that there exist cycles 
$C_1, \ldots, C_{\ell}$ in 
$\B_1 \cup \ldots \cup \B_k$ with 
$\ex(C_1) \oplus \ldots \oplus \ex(C_{\ell}) = 0$. 
For each $h$ let $R_h$ be the set of indices $r$ such that $C_r \in G_h$.
Since the only vertices that appear in more than one subgraph $G_h$ are the two vertices $u$ and $v$ and since the sum of cylces forms a disjoint union of cycles, $\sum_{r \in R_h}{C_r}=0$ must hold. 
This implies $R_h = \emptyset$ for all $h$ and shows the independence of the cycles in $\EE$.

Since the expansion of a cycle $C \in \B_h$ does not change its total cost, the total weight of the cycles in $\EE$ equals the total weight of the cycles in $\B_1 \cup \ldots \cup \B_k$. 
We need to show that there is no cycle basis of $G$ that has strictly smaller cost.
By Lemma~\ref{lem:lscdecomp} we know that, if we decompose $G$ along the lex shortest path $\lsp_{G}(u,v)$ between $u$ and $v$, i.e., if we apply $\decomp(G,u,v)$, then the set $\EE^*$ as defined in Lemma~\ref{lem:lscdecomp} forms a minimum cycle basis. 
We show that the cycles in $\EE^*$ are of the same total weight as the cycles in $\EE$. 
If $\{u,v\} \in E$, this is trivially true as the decompositions $\decomp(G,u,v)$ and $\de(G,u,v)$ are identical.
Furthermore, if $i:=i(u,v)=j(u,v)=:j$, then $\short(u,v)=\lsp(u,v)$ since by Observation~\ref{cor:nocycle} this is the only path between $u$ and $v$ that is contained in $G^*_i$. That is, also in case $i=j$ the two decompositions are identical, implying in particular the minimality of the weights in $\B_1 \cup \ldots \cup \B_k$.
We thus assume in the following that $\{u,v\} \notin E$ and that $i \neq j$. 
We define a bijective mapping 
\begin{align*}
\psi:\LSC(G^*_1) \cup \ldots \cup \LSC(G^*_k) \rightarrow \LSC(G_1) \cup \ldots \cup \LSC(G_k)
\end{align*}
such that $w(\psi(C))=w(C)$ for all $C \in \LSC(G^*_1) \cup \ldots \cup \LSC(G^*_k)$. 
From this we get the statement by evoking again Lemma~\ref{lem:LSCMCB} which tells us that for each $h$ the sum of the weights of the cycles in $\LSC(G_h)$ equals the sum of the weights of the cycles in $\B_h$.
For $h \notin \{i,j\}$, $G_h = G^*_h \setminus \{ \blue(u,v) \} \cup \green(u,v)$ holds. 
For $C \in \LSC(G^*_h)$ we can therefore set $\psi(C):=C\setminus\{ \blue(u,v) \} \cup \green(u,v)$ if $\blue(u,v) \in C$, and $\psi(C):=C$ otherwise.
As we have mentioned in Observation~\ref{cor:nocycle}, the only path between $u$ and $v$ in $G^*_i$ is $\lsp(u,v)$. In particular, there is no cycle that contains both $u$ and $v$. It is thus easy to see that $\LSC(G^*_i) \subseteq \LSC(G_i)$. Set $\psi(C):=C$ for all $C \in \LSC(G^*_i)$. 
Since there is one more edge in $G_i$ as there is in $G^*_i$, the number of cycles in $\LSC(G_i)$ is by one larger than the number of cycles in $\LSC(G^*_i)$. Note that $D:=\{\green(u,v)\} \cup \lsp_{G}(u,v)$ is a lex short cycle in $G_i$ that is not in the image $\{\psi(C) \mid C \in \LSC(G^*_i)\}$. The cost of $D$ is $2 w(\lsp_{G}(u,v))$. 
Note that the situation is symmetric for $G^*_j$ and $G_j$. We thus set $\psi(D^*):=D$ for $D^*:=\{\blue(u,v)\} \cup \short(u,v)$ and we set $\psi(C):=C\setminus\{\blue(u,v)\} \cup \short(u,v)$ for all other cycles in $\LSC(G^*_j)$. As the cost $w(D^*)$ equals $2 w(\lsp_{G}(u,v))$ as well, the claim follows.
\qed
\end{proof}

\section{Computing an MCB in Weighted Partial 2-Trees}
\label{sec:main}

\begin{algorithm2e}[t]
Compute a suitable tree decomposition $T$ of $G$\label{line:suitableTreeDecomp}\;
Find the set $L$ of long edges in $G$\label{line:findLongEdges}\;
$E \leftarrow E \setminus L$\label{line:deleteLong}\;
\For{each internal bag $Y_1 \in T$ (in any order) and every $u,v \in Y_1$}
{
Let $Y_2, \ldots Y_k$ be the children of $Y_1$ such that for $2 \leq i \leq k, Y_1 \cap Y_{i} = \{u,v\}$\;
\If{$k \geq 3$}
	{
		\lFor{$2 \leq i\leq k$}
		{delete the link $(Y_1, Y_i)$\label{line:deletinglinks}\;}
		\If{$\{u,v \} \notin E$}{
				Compute the weight $w(P)$ of a shortest path $P$ between $u$ and $v$\label{line:distancequery}\;
				Find an intermediate vertex $y$ of $P$ and a bag $B$ containing $y$\label{line:markbluestart}\;
				Compute $j$ such that either $j \in \{2,\ldots,k\}$ and the subtree rooted at $Y_j$ contains $B$ or $j=1$ otherwise (indicating that $B$ is in the subtree containing $Y_1$)\label{line:intermediate}\label{line:markblueend}\;
				\For{$1 \leq h \neq j \leq k$}
					{Add the new edge $\green(u,v)$ to $Y_h$ and assign to it weight $w(P)$\;}
					}
				}
	}
Let $\tilde{T}_1, \ldots, \tilde{T}_r$ be the connected components of $T$\;
Obtain the outerplanar graphs $\tilde{G}_1, \ldots, \tilde{G}_r$ that correspond to $\tilde{T}_1, \ldots, \tilde{T}_r$\label{line:Gi}\;
\label{line:findInternalFaces}Compute $\LSC(\tilde{G}_1), \ldots, \LSC(\tilde{G}_r)$ using~\cite{LiuL10}\;
Output (in an implicit or explicit representation): $\MCB(G) = \ex(\LSC(\tilde{G}_1)) \uplus \cdots \uplus \ex(\LSC(\tilde{G}_r)) \uplus \{\{e\} \cup \short(u,v) \mid e=\{u,v\} \in L\}$\label{line:output}\; 
\caption{A linear time algorithm to compute a minimum cycle basis of a weighted 2-connected partial 2-tree $G$}
\label{alg:main}
\end{algorithm2e}

As we now have all the technical tools at hand, we can finally give a detailed description of our algorithm, whose pseudo-code can be found in Algorithm~\ref{alg:main}.

We first compute a suitable tree decomposition $T$ of the $2$-connected weighted partial 2-tree $G$. According to Lemma~\ref{lem:makesuitable}, this can be done in linear time. We then compute the set $L$ of long edges, i.e., the edges whose length is greater than the length of a shortest path between their two endpoints (cf.\ Lemma~\ref{lem:tight}). By using distance queries between the endpoints of every edge, the data structure of Lemma~\ref{lem:DistanceOracleBoost} allows to do this in linear time. We delete the edges in $L$ and consider thus $G \setminus L$ until Line~\ref{line:Gi} of Algorithm~\ref{alg:main} is reached; clearly, $T$ is still a suitable tree decomposition of $G \setminus L$.

As described in the high-level overview, we decompose the graph iteratively into outerplanar graphs along $K_{2,3}$-subdivisions (cf.\ Section~\ref{sec:highlevel}). Algorithmically, Corollary~\ref{cor:graphTreeDecomp2} allows us to detect efficiently whether the current graph contains a $K_{2,3}$-subdivision: We just have to check whether $T$ contains at least three bags each of which contains the same two vertices $u$ and $v$. Since $T$ is suitable, this can be done by fixing every inner vertex $Y_1$ of $T$ (in any order) and counting the number $k$ of children of $Y_1$ whose links are labeled identically, say with $\{u,v\}$ (cf.\ the first three lines of the main loop of Algorithm~\ref{alg:main}). If $k \geq 3$, we have identified a $K_{2,k}$-subdivision by Lemma~\ref{obs:graphTreeDecomp1}.

If $\{u,v\} \in E$ (the existence of such an edge can be efficiently looked up by a table of size $O(n)$ with Lemma~\ref{lem:DistanceOracle}), we can simply decompose the graph into the parts defined by $\de$ by deleting $k$ links in $T$ (cf.\ Line~\ref{line:deletinglinks}). Otherwise, we additionally fix a shortest path $P$ between $u$ and $v$ and augment all parts except the one containing $P$ with a green edge that replaces $P$. For this purpose we have to find the weight of $P$ and the part that contains $P$. The first can be done in constant time by using a distance query of Lemma~\ref{lem:DistanceOracleBoost} (cf.\ Line~\ref{line:distancequery}). The latter is computed by identifying the subtree of the tree decomposition (the one after deleting the links) that contains an intermediate vertex $y$ of $P$ (cf.\ Line~\ref{line:markblueend}). Such a vertex can be computed in constant time, using once more the data structure of Lemma~\ref{lem:DistanceOracleBoost}.

Finally, we end up with several components $\tilde{T}_1,\ldots,\tilde{T}_r$ of the original tree decomposition $T$; these are easy to find in linear time, e.g., by depth-first search. We can compute the graphs $\tilde{G}_1,\ldots, \tilde{G}_r$ that are represented by these tree decompositions in linear total time by simply collecting the vertices and edges in all bags. Note that the total number of edges in $\tilde{G}_1,\ldots, \tilde{G}_r$ is still in $O(n)$, as we add at most $\degree(V_1)$ new green edges for each bag $V_1$, where $\degree(V_1)$ is the degree of bag $V_1$ in $T$. Every $G_i$ is outerplanar; thus, we can compute the $\LSC$ of every $G_i$ in linear time (cf.\ Theorem~\ref{thm:LiuL}). According to Lemmata~\ref{lem:tight} and~\ref{thm:mapping1}, the output in Line~\ref{line:output} is then a minimum cycle basis of $G$.

This concludes the first part of our main result, Theorem~\ref{thm:main}. It remains to clarify how $\MCB(G)$ is represented in the output. 
For an implicit representation, we store $G$, $L$, $\LSC(\tilde{G}_1), \ldots, \LSC(\tilde{G}_r)$ and a trace of the main loop of Algorithm~\ref{alg:main}. 
Clearly the space consumption is in $O(n)$. For every long edge $e=\{u,v\}$ in $L$, we can compute an arbitrary shortest path between $u$ and $v$ in $G$ in time proportional to its length; as the choice of this path does not matter due to Lemma~\ref{lem:tight}, this will complete every long edge to a cycle of an $\MCB(G)$. 
The trace stores every decision that was made in the decomposition $\de(G \setminus L)$. It thus allows to reconstruct the whole decomposition in linear time, as Algorithm~\ref{alg:main} takes linear time.

In particular, we can identify for every decomposition step that is performed on a graph $H$, the part $H_{j(u,v)}$ of $\de(H,u,v)$ in which we had chosen the shortest path $\short(u,v)$ between $u$ and $v$, whose length we computed. An explicit representation of $\ex(\LSC(\tilde{G}_1)) \uplus \cdots \uplus \ex(\LSC(\tilde{G}_r))$ is then computed by constructing the shortest path $\short(u,v)$ for every decomposition step explicitly via Lemma~\ref{lem:DistanceOracleBoost} (in contrast to computing only its length and an intermediate vertex, as in the original decomposition).

According to Lemma~\ref{lem:DistanceOracleBoost}, computing these shortest paths takes time proportional to the number of edges in these paths.
This gives the desired running time of $O(\size(\MCB(G)))$, where $\size(\MCB(G))$ is the number of edges in $\MCB(G)$ counted according to their multiplicity.

\section{Discussion}\label{sec:conclusions}
We have shown that an implicit representation of a minimum cycle basis of a weighted partial 2-tree can be computed in linear time. 
It remains a challenging question if our result can be extended to partial $k$-trees for $k>2$. We remark that it was noted in~\cite{CTW12} that already for partial 3-trees the set of lex short cycles do not necessarily form a minimum cycle basis.
Since in particular the proof of Theorem~\ref{thm:mapping1} is based on this, extending our result to partial 3-trees may therefore require substantially new ideas.

\subsubsection*{Acknowledgments.}
We would like to thank Geevarghese Philip for pointing us to Lemma~\ref{lem:Diestel}, which significantly simplified our proof of the runtime bound. We also thank the anonymous reviewers for providing their feedback which has helped us to improve the presentation of our work.

Carola Doerr gratefully acknowledges support from the Alexander von Humboldt Foundation and the Agence Nationale de la Recherche (project ANR-09-JCJC-0067-01).

G.\ Ramakrishna would like to thank his adviser N.S.\ Narayanaswamy for fruitful discussions during the early stages of this work. He would also like to thank Kurt Mehlhorn for providing him the opportunity to do an internship at the Max Planck Institute for Informatics.
}


\end{document}